\documentclass[10pt,twocolumn]{IEEEtran}

\usepackage{amsmath, mathrsfs}
\usepackage{amsfonts}
\usepackage{amssymb}
\usepackage{graphicx}
\usepackage{color}
\usepackage{multirow}
\usepackage{graphicx}

\usepackage{algorithm}
\usepackage{algpseudocode}
\usepackage{pifont}
\usepackage{varwidth}

\newcommand{\myhash}{%
  {\settoheight{\dimen0}{C}\kern-.05em\, \resizebox{!}{\dimen0}{\raisebox{\depth}{\#}}}}

\usepackage{caption}
\usepackage{subcaption}

\usepackage[numbers,sort&compress]{natbib}

\usepackage{multicol}


\def\mindex#1{\index{#1}}

\def\sq{\hbox{\rlap{$\sqcap$}$\sqcup$}}
\def\qed{\ifmmode\sq\else{\unskip\nobreak\hfil
\penalty50\hskip1em\null\nobreak\hfil\sq
\parfillskip=0pt\finalhyphendemerits=0\endgraf}\fi\medskip}

\long\def\defbox#1{\framebox[.9\hsize][c]{\parbox{.85\hsize}{%
\parindent=0pt
\baselineskip=12pt plus .1pt      % STYLE
\parskip=6pt plus 1.5pt minus 1pt % CHANGES
 #1}}}

\long\def\beginbox#1\endbox{\subsection*{}%
\hbox{\hspace{.05\hsize}\defbox{\medskip#1\bigskip}}%
\subsection*{}}

\def\endbox{}

\def\diag{{\text{diag}}}

\def\tr{\mathsf{Tr}}

\newsavebox{\junk}
\savebox{\junk}[1.6mm]{\hbox{$|\!|\!|$}}

\def\argmin{\mathop{\rm arg\, min}}

\def\argmax{\mathop{\rm arg\, max}}

\newcommand{\field}[1]{\mathbb{#1}}

\def\Re{\field{R}}

\def\ind{\field{I}}

\def\ZZ{\field{Z}}

\def\bC{{\mathbb C}}

\def\bE{{\mathbb E}}

\def\bH{{\mathbb H}}

\def\bP{{\mathbb P}}

\def\bR{{\mathbb R}}

\def\bT{{\mathbb T}}

\def\bV{{\mathbb V}}

\def\bfA{{\bf A}}

\def\bfC{{\bf C}}
\def\bfD{{\bf D}}

\def\bfF{{\bf F}}

\def\bfI{{\bf I}}

\def\bfN{{\bf N}}

\def\bfP{{\bf P}}

\def\bfS{{\bf S}}

\def\bfU{{\bf U}}

\def\bfX{{\bf X}}

\def\bfa{{\bf a}}

\def\bfc{{\bf c}}

\def\bfe{{\bf e}}
\def\bff{{\bf f}}

\def\bfn{{\bf n}}

\def\bfr{{\bf r}}

\def\bfu{{\bf u}}
\def\bfv{{\bf v}}
\def\bfw{{\bf w}}
\def\bfx{{\bf x}}
\def\bfy{{\bf y}}

\def\frkD{{\mathfrak{D}}}

\def\frkF{{\mathfrak{F}}}

\def\frka{{\mathfrak a}}

\def\frkc{{\mathfrak c}}

\def\frkr{{\mathfrak r}}

\def\sfH{{\sf H}}

\def\sfa{{\sf a}}

\def\sfr{{\sf r}}

\def\bfmath#1{{\mathchoice{\mbox{\boldmath$#1$}}%
{\mbox{\boldmath$#1$}}%
{\mbox{\boldmath$\scriptstyle#1$}}%
{\mbox{\boldmath$\scriptscriptstyle#1$}}}}

\def\bfmY{\bfmath{Y}}

\def\bfmhhaY{\bfmath{\hhaY}} %\widehat{\widehat{Y}}}}
\def\bfmhhaY{\hbox to 0pt{$\widehat{\bfmY}$\hss}\widehat{\phantom{\raise 1.25pt\hbox{$\bfmY$}}}}

\def\til={{\widetilde =}}

\def\clA{{\cal A}}
\def\clB{{\cal B}}
\def\clC{{\cal C}}

\def\clF{{\cal F}}

\def\clN{{\cal N}}

\def\clX{{\cal X}}

 \def\FRAC#1#2#3{\genfrac{}{}{}{#1}{#2}{#3}}

\def\ddtp{{\mathchoice{\FRAC{1}{d^{\hbox to 2pt{\rm\tiny +\hss}}}{dt}}%
{\FRAC{1}{d^{\hbox to 2pt{\rm\tiny +\hss}}}{dt}}%
{\FRAC{3}{d^{\hbox to 2pt{\rm\tiny +\hss}}}{dt}}%
{\FRAC{3}{d^{\hbox to 2pt{\rm\tiny +\hss}}}{dt}}}}

\def\average#1,#2,{{1\over #2} \sum_{#1}^{#2}}

\def\eye(#1){{\bf(#1)}\quad}

\def\var{{\bV\sfa\sfr}}
\newtheorem{theorem}{{\bf Theorem}}[section]

\newtheorem{proposition}[theorem]{{\bf Proposition}}

\def\eq#1/{(\ref{e:#1})}

\newcommand{\inp}[2]{{\langle #1, #2 \rangle}}
\newcommand{\inpr}[2]{{\langle #1, #2 \rangle}_\bR}

\newcommand{\beqn}[1]{\notes{#1}%
\begin{eqnarray} \elabel{#1}}

\newcommand{\eeqn}{\end{eqnarray} }

\newcommand{\beq}[1]{\notes{#1}%
\begin{equation}\elabel{#1}}

\newcommand{\eeq}{\end{equation}}

\def\bdes{\begin{description}}
\def\edes{\end{description}}

\newcounter{rmnum}

\newcounter{anum}

{\end{list}}

\def\ass(#1:#2){(#1\ref{#1:#2})}

\def\ritem#1{
\item[{\sf \ass(\current_model:#1)}]
}

\newenvironment{recall-ass}[1]{%
\begin{description}
\def\current_model{#1}}{
\end{description}
}

\usepackage{tikz}
\usepackage{pgfplots}
\pgfplotsset{compat=newest}
\usetikzlibrary{patterns}
\usetikzlibrary{positioning}
\usetikzlibrary{datavisualization}
\usetikzlibrary{datavisualization.formats.functions}
\usetikzlibrary{backgrounds}
\usetikzlibrary{shapes,snakes}

\def\lcav{{L_{\text{cav}}}}

\renewcommand\thesection{\arabic{section}}
\renewcommand\thesubsection{\thesection.\arabic{subsection}}
\renewcommand\thesubsubsection{\thesubsection.\arabic{subsubsection}}

\renewcommand\thesectiondis{\arabic{section}}
\renewcommand\thesubsectiondis{\thesectiondis.\arabic{subsection}}
\renewcommand\thesubsubsectiondis{\thesubsectiondis.\arabic{subsubsection}}

\def\herm{{\sfH}}
\def\sdet{{\mathsf{det}}}
\def\snr{{\mathsf{snr}}}

\newcommand{\range}[2]{{\text{$#1$\,:\,$#2$}}}

\def\cg{{\clC\clN}}

\long\def\comment#1{}

\newfont{\bbb}{msbm10 scaled 700}

\newfont{\bb}{msbm10 scaled 1100}

\newcommand{\xv}{{\bf x}}

\newcommand{\Id}{{\bf I}}

\newcommand{\Xm}{{\bf X}}

\newcommand{\Xc}{{\cal X}}

\newcommand{\Lambdam}{\hbox{\boldmath$\Lambda$}}

\newcommand{\cov}{{\hbox{cov}}}

\renewcommand{\Re}{{\rm Re}}
\renewcommand{\Im}{{\rm Im}}

\newcommand{\transp}{{\sf T}}

\begin{document}

\title{An Active-Sensing Approach to Channel Vector Subspace Estimation in
mm-Wave Massive MIMO Systems}
\author{\IEEEauthorblockN{Saeid Haghighatshoar\IEEEauthorrefmark{1},
Giuseppe Caire\IEEEauthorrefmark{1}}\\
\IEEEauthorblockA{\IEEEauthorrefmark{1}\IEEEauthorblockA{Communications and Information Theory Group, Technische Universit\"{a}t Berlin}}\\
Emails: saeid.haghighatshoar@tu-berlin.de, caire@tu-berlin.de}

\maketitle
\begin{abstract}
Millimeter-wave (mm-Wave) cellular systems are a promising option for a very high data rate communication because of the large bandwidth available at mm-Wave frequencies. Due to the large path-loss exponent in the mm-Wave range of the spectrum, directional beamforming with a large antenna gain is necessary at the transmitter, the receiver or both for capturing sufficient signal power.
This in turn implies that fast and robust channel estimation plays a central role in systems performance since without a reliable estimate of the channel state the received \textit{signal-to-noise ratio} (SNR) would be much lower than the minimum necessary for a reliable communication.

In this paper, we mainly focus on single-antenna users and a multi-antenna base-station. We propose an adaptive sampling scheme to speed up the user's signal subspace estimation. In our scheme, the beamforming vector for taking every new sample is adaptively selected based on all the previous beamforming vectors and the resulting output observations. We apply the theory of optimal design of experiments in statistics to design an adaptive algorithm for estimating the signal subspace of each user. The resulting subspace estimates for different users can be exploited to efficiently communicate to the users and to manage the interference. We cast our proposed  algorithm as low-complexity optimization problems, and illustrate its efficiency via numerical simulations.

\end{abstract}

\section{Introduction}
Consider a mm-Wave cellular system consisting of a multi-antenna base-station with $M$ antennas serving $K$ single-antenna users. For simplicity, we focus here on a flat-fading channel in which the bandwidth of the signal is less than the channel's coherence bandwidth\footnote{For example, our model may be regarded as the channel resulting from a single subcarrier of an OFDM system.}. Due to high path-loss exponent in the mm-Wave frequencies, a high beamforming gain at the base-station is necessary in order to deliver/receive sufficient signal power to/from the users \cite{pi2011introduction, rappaport2013millimeter}. Traditionally, the design of 
beamforming matrices is based on complete channel information, which is difficult to have in mm-Wave spectrum due to the small {\em signal-to-noise ratio} (SNR) before beamforming. This in turn indicates that a fast and robust acquisition stage plays a crucial role for establishing a reliable communication link to the users.   

Adaptive beamwidth beamforming algorithms with multi-stage codebooks have been recently proposed by which the transmitter and receiver jointly design their beamforming vectors \cite{wang2009beam, chen2011multi, hur2013millimeter, alkhateeb2014channel}. In short, the proposed multi-stage (multi-resolution) beamforming algorithms can be seen as different variants of binary search or bisection algorithm over all possible {\em angle of arrivals} (AoA) of the channel between the transmitter and the receiver. The search typically starts with two wide beams each scanning half the range of possible AoAs. This provides a rough estimate of the location of channel's dominant AoAs in either the first of the second half. The resulting estimate is further refined by more and more narrow beams until a good estimate of the location of AoAs is found. The main issue of all these search algorithms is that they look for the dominant AoAs rather than the signal subspace (the subspace containing a significant amount of signal power). In particular, the binary search applied by such algorithms becomes less efficient when there are several significant AoAs rather than one. 

In this paper, we mainly focus on single-antenna users and a multi-antenna base station. We formulate the acquisition problem for each user in terms of estimating the user's dominant signal subspace rather than searching  over the AoAs of its communication channel to the base-station. Further, we use the theory of optimal experiment design in statistics to develop an adaptive beamforming and subspace estimation algorithm. We assume that a stream of independent training samples (observations) of each user's channel vector is available at the base station via sequential training. For every new training sample, our algorithm updates its estimate of the user's signal subspace by exploiting all the available measurements. The resulting estimate of the subspace is in turn used to find an optimal beamforming vector for taking the next measurement in the next training snapshot. For simplicity we focus on a 1-dim beamformer that exploits only one RF chain for sampling the whole $M$-dim array signal. We cast the optimal design of this beamforming vector, and subspace estimation as low-complexity optimization problems that can be efficiently solved. We also provide numerical simulations to asses the performance of our proposed algorithm.

\subsection{Notations} \label{sec:notation}
%
%We do not distinguish, in terms of notation, between deterministic quantities and random variables, since the meaning of the various symbols
%is clear from the context. We use $\PP(\cdot)$ to denote a generic probability measure for the probability space that includes all the random variables that appear in the event inside the parentheses. 
%We use $M, m, T$ for the array size, the dimension of the observation (sketch) and the number of training slots. 
We denote vectors by boldface small letters (e.g., $\xv$), matrices by boldface capital letters (e.g., $\Xm$), scalar constant by 
non-boldface letters (e.g., $x$ or $X$), and sets by calligraphic letters (e.g., $\Xc$). The $i$-th element of a vector $\xv$
and the $(i,j)$-th element of a matrix $\Xm$ will be denoted by $[\xv]_i$ and $[\Xm]_{i,j}$ respectively. 
%%%%%%%%%%%%%%%% introduce matlab notation ????????????????????
%We denote by $\Cc\Nc(\muv, \mathbf{\Sigma})$ the complex circularly symmetric multivariate Gaussian distribution with mean vector $\muv$ and covariance matrix $\mathbf{\Sigma}$, and we use $\sim$ to indicate ``distributed as''. 
Throughout the paper, the output of an optimization algorithm $\argmin_{x}  f(x)$ is denoted by $x^*$. 
We denote the Hermitian and the transpose of a matrix  $\bfX$ by $\bfX^\herm$ and $\Xm^\transp$, respectively. 
The same notation is used for vectors and scalars. 
We use $\bT_+$ for the space of Hermitian semi-definite Toeplitz matrices. For an $\bfx \in \bC^M$, we denote by $\bT(\bfx)$ a Hermitian Toeplitz matrix whose first column is $\bfx$.
We always use $\Id$ for the identity matrix, where the dimensions may be explicitly indicated for the sake of clarity (e.g., $\Id_p$ denotes
the $p \times p$ identity matrix).  
%We define $\bH(M,p)=\{\bfU_{M\times p} \in \bC^{M\times p}: \bfU^\herm \bfU=\Id_p\}$ as the set of tall unitary matrices 
%of dimension $M\times p$. For matrices and vectors of appropriate dimensions, 
%we define the inner product $\inp{\bfK}{\bfL}=\tr(\bfK \bfL^\herm)$ and  the induced norm 
%$\|\bfK\|=\sqrt{\inp{\bfK}{\bfK}}$, also known as {\em Frobenius norm} for matrices. 
For an integer $k \in \ZZ$, we use the shorthand notation $[k]$ 
for the set of non-negative integers $\{0,1,2, \dots, k-1\}$, where the set is empty if $k < 0$. 
%We denote the cardinality of a set $\Xc$ by $|\Xc|$.

\section{Preliminaries}\label{sec:rel_work}
In this section, we briefly overview the theory of optimal experiment design in statistics and highlight its connection 
with the channel estimation problem that we address in this paper. In the traditional statistical analysis, the statistician deals with a collection of observations $X$, belonging to some observation space $\clX$, whose distribution is governed by the state of the nature $\theta \in \Theta$. Although the state $\theta \in \Theta$ is not known to the statistician, it is known that it can be described 
by a family of probability distributions $\{\bP_\theta\}_{\theta \in \Theta}$ over $\clX$. In particular, it is always assumed that, the family 
$\{\bP_\theta\}_{\theta \in \Theta}$ is fully known to the statistician. 
When the nature choses the state $\theta$, the role of the statistician is to infer its value from the observation $X$ (or to estimate some sufficient statistics $g : \clX \rightarrow \bR$ thereof). This is essentially the classical decision theoretic definition of a statistical experiment $(\Theta, \{\bP_\theta\}_{\theta \in \Theta}, \clX)$, which can be found in textbooks such as \cite{lehmann1998theory, lehmann2006testing, berger2013statistical}.

A richer and at the same time more challenging problem is the {\em optimal experiment design}. In this case, the distribution of the data $X \in \clX$ depends not only on the state of the nature $\theta$ but also on a parameter $\beta \in \clB$ that can be controlled by the statistician. The idea is that by suitably manipulating
$\beta$, the statistician is able to take more informative measurements of the parameter $\theta$. This is realized via an adaptive or a nonadaptive sampling scheme. 
In a nonadaptive scheme, a collection of design parameters $\{\beta_t\}_{t=1}^T$ for the whole duration of the experiment $T$ is a priori selected. 
In an adaptive procedure, on the other hand, having the collection of observations $\{X_1,X_2, \dots, X_{t-1}\}$, the statistician selects the parameter $\beta_t$ for time $t$ to take the next random sample $X_t\in \clX$ according to the distribution $\bP_{\theta, \beta_t}$. The difference with the nonadaptive case is that now the design parameter $\beta_t \in \clB$ can be measurable function of the whole information $\clF_{t-1}=\sigma(X_1, X_2, \dots, X_{t-1})$ causally available at time $t-1$. Usually the goal is to minimize the number of samples for a given  performance guarantee or to improve the performance for a fixed number of samples. 

The simplest and the well-analyzed case of the experiment design is the linear case in which 
\begin{align}\label{lin_resp}
\bE_{\theta, \beta}[X]=f(\beta)^\transp \theta, \ \var_{\theta, \beta}[X]=\sigma^2,
\end{align}
where it is assumed that $\theta \in \Theta \subset \bR^n$ and the observation mean is a linear function of $\theta$ parametrized with a possibly nonlinear function $f:\clB \to \bR^n$ of the design parameter $\beta \in \clB$. It is also assumed the variance of the observations $\sigma^2$ is a constant independent of $\theta$. For a given choice of the design parameters $\{\beta_t\}_{t=1}^T$, it is not difficult to check that the covariance of the {\em best linear unbiased estimator} (BLUE) of $\theta$ given the observations $\{X_t\}_{t=1}^T$, with $X_t \sim \bP_{\theta, \beta_t}$, is given by the {\em positive semi-definite} (PSD) matrix $\sigma^2 (\bfF \bfF^\transp)^{-1}$ where $\bfF=[f(\beta_1), \dots, f(\beta_T)]$ is an $n\times T$ matrix consisting of the design vectors $f(\beta_t)$, $t\in[T]$. 
There are different optimality criteria that have been studied in the literature such as D-optimality and A-optimality (cf. \cite{fedorov1972theory, fedorov2013optimal} and the references therein). For example, if one is interested in finding the optimal design minimizing $\sum_{i=1}^n \var(\widehat{\theta}_i)$, this is equivalent to minimizing $\tr[(\bfF \bfF^\transp)^{-1}]$, which is an example of A-optimal design obtained by minimizing a cost function $\Psi_\bfA(\bfF)=\tr[\bfA (\bfF \bfF^\transp)^{-1}]$, parametrized by a suitable matrix $\bfA$, e.g. $\bfA=\bfI$ in this case. In a D-optimal design, on the other hand, one is interested in minimizing  $\Psi(\bfF)=\sdet[(\bfF \bfF^\transp)^{-1}]$. The intuition is that, if the estimation error turns out to be asymptotically Gaussian, then $\sdet[(\bfF \bfF^\transp)^{-1}]$ gives the volume of the error ellipsoid around the true parameter $\theta$ obtained via the level sets of the Gaussian probability distribution, thus, it is a reasonable metric to be minimized. It has been shown that both criteria for the optimal design can be cast as convex optimization problems that can be efficiently solved \cite{fedorov1972theory, fedorov2013optimal}. 

The linear regime given by \eqref{lin_resp}, is probably the simplest case of experiment design in which the optimal experiment for a linear estimator can be designed nonadaptively before running the experiment. In particular, the optimal design does not need the knowledge of the true parameter $\theta$, mainly because the variance term $\var_{\theta, \beta}(X)$ does not depend on $\theta$. In practice, and in particular in the problem that we address in this paper, the second oder statistics depends on $\theta$, and typically some extra information about the distribution of $X$ beyond the first and second order statistics is available that can be further exploited. This implies that an optimal design is not in general possible in the nonadaptive case. One typically needs a sequential procedure in which an estimate $\widehat{\theta}_t$ of $\theta$ is obtained via the observed data $\{X_\ell\}_{\ell=1}^t$. This new estimate is in turn used to find an optimal design $\beta^*_{t+1}$ for sampling the next observation $X_{t+1}$. In this paper, we apply such a sequential procedure to estimate the signal subspace of each user, where in each step we apply a sequential D-optimal design to find the optimal beamforming vector for taking the next observation.

\section{Statement of the Problem}
\subsection{Signal Model}\label{sec:sig_model}
We consider a simple propagation model for the wireless scattering channel in which the transmission between a single-antenna user and the multiple-antenna base-station occurs through 
$p$ scatterers
(see Fig.~\ref{fig:sc_channel}). Due to the sparse scattering in mm-Wave spectrum and the large antenna size $M$ used in massive MIMO, we typically have $p\ll M$ \cite{adhikary2014joint}. We assume that the base-station is equipped with a {\em Uniform Linear Array} (ULA), with spacing $d=\frac{\lambda}{2 \sin(\theta_m)}$ between its elements, and scans the angular range $[-\theta_m, \theta_m]$ for some $\theta_m \in (0,\pi/2)$. 
\begin{figure}[h]
\centering
\begin{tikzpicture}[scale=0.55, every node/.style={scale=0.8}]
%styles
\tikzstyle{st_axis}=[thick,dashed,->, blue];
\tikzstyle{st_array}=[line width=2, rotate=45];
\tikzstyle{st_arr_elem}=[line width=1, red];

% draw base line
\draw[red] (0,0) -- (1.5,0);

%axis
\draw[thick] (0,-2.4) -- (0,2.4);
\draw[thick, ->] (0,2.4) -- (-0.1,2.4) -- (0.1, 2.6) -- (0,2.6) -- (0,3.5);

% array elements
\draw (0,-2) node{$\bullet$} node[left] {$0$};
\draw (0,-1) node{$\bullet$} node[left] {$d$};
\draw (0,0) node{$\bullet$} node[left] {$2d$};
\node (BS) at (0,0) {$$};
\draw (0,1) node{$\bullet$} node[left] {$3d$};
\draw (0,2) node{$\bullet$} node[left] {$4d$};
\draw (0,3) node{$\bullet$} node[left] {$(M-1)d$};

% user mobile
\draw[red, thick] (8,-0.25) -- (8,0.25) node[black, above=0.5cm] {User};
\node (U) at (8,0) {$$};
\draw[red, thick] (8,0.25) -- +(45:0.3);
\draw[red, thick] (8,0.25) -- +(135:0.3);

% put scatterer
\node[draw, circle, shade] (A) at (4,3) {$$};
\node[draw, circle, shade] (B) at (3.5,1) {$$};
\node[draw, circle, shade] (C) at (4.5,-1) {$$};
\node[draw, circle, shade] (D) at (4,-3) {$$};
\node at (4,0) {$\vdots$};
\draw[very thick, brown, dashed] (4,0) ellipse (1 and 3.5);

% connect the scatterere
\foreach \x in {A,B,C,D}{
	\draw[->] (U) -- (\x) -- (BS);
}

% write channel
\draw (A) node[above=0.5cm] {Scattering Channel};
\draw[blue, thick,->] (0:1) arc(0:38:1)  node[above] {$\theta_i$};
\end{tikzpicture}
\caption{{\footnotesize Scattering channel with discrete angle of arrivals.}}
\label{fig:sc_channel}
\end{figure}
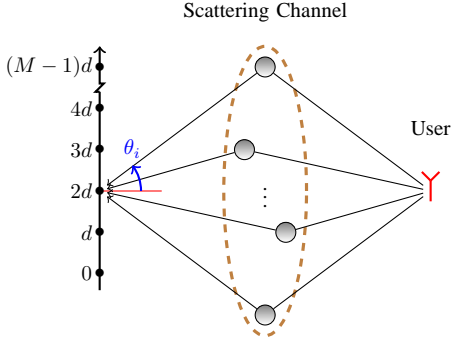
Without loss of generality, we focus on a single user. One snapshot of the received signal is given by
\begin{equation}
\bfy=\sum_{\ell=1}^p \bfa(\theta_\ell) w_\ell\,  x + \bfn, 
\end{equation}
where $x$ is the transmitted (training) symbol from the user, $w_\ell \sim \cg(0, \sigma_\ell^2)$ is the channel gain 
of the $\ell$-th multipath component, $\bfn \sim \cg(0, \sigma^2 \Id_M)$ is the additive white Gaussian noise of the receiver antenna, and where $\bfa(\theta) \in \bC^M$ is the array response at AoA $\theta$, whose $k$-th component is given by 
\begin{align}
[\bfa(\theta)]_k= \exp(j k\frac{2\pi d \sin(\theta)}{\lambda} k)=\exp(j k\pi \frac{\sin(\theta)}{\sin(\theta_m)}).
\end{align}
According to the WSSUS model, the channel gains for different paths $\{w_\ell\}_{\ell=1}^p$ 
are uncorrelated, and since they are (jointly) Gaussian, they are mutually statistically independent and independent of the receiver noise $\bfn$. 
%We define the array response vector 
%$\bfa(\theta) \in \bC^M$ at AoA $\theta$, with $k$-th component given by 
%$a_k(\theta)= \exp(j \frac{2\pi d \sin(\theta)}{\lambda} k)$, $k \in [M]$, where $\lambda$ is the wavelength. 
%{\RED As we will mainly work with the standard linear array ($d=\frac{\lambda}{2}$), we have $a_k(\theta)= \exp(j \pi \sin(\theta) k)$. Defining the wave-number in the $z$-direction (parallel to the array) by $u=\pi \sin(\theta)$, this can also be written as $a_k(u)= \exp(j u k)$, which shows that the array vector is a harmonic function of $u$. [FOR THE SAKE OF GENERALITY AND SINCE THIS IS NOT REALLY CRITICAL, I WOULD KEEP THE DEPENDENCY ON $\theta$ SUCH THAT THERE IS NO AMBIGUITY OR FOLDING, 
%AND AVOID RESTRICTING TO HALF-LAMBDA]}
Without loss of generality, we suppose $x=1$ in all training snapshots. Letting 
$\bfA = [\bfa(\theta_1), \bfa(\theta_2), \dots,  \bfa(\theta_p)]$, we have 
\begin{align}\label{eq:disc_ch_mod}
\bfy [t]=\bfA \bfw [t] + \bfn [t], \;\;\; t\in[T],
\end{align}
where $\bfw [t] = (w_1 [t], w_2 [t], \dots, w_p [t])^\transp$ for different $t$ are statistically independent. Also,  we assume that the AoAs $\{\theta_\ell\}_{\ell=1}^p$ remain invariant over a training period of length $T$.
%Gathering the received signal's snapshots, the channel coefficient snapshots, and the receiver noise snapshots 
%into matrices $\bfY =[\bfy(1), \dots ,\bfy(T)]$
%$\bfW = [\bfw(1), \dots ,\bfw(T)]$, and $\bfN = [\bfn(1), \dots ,\bfn(T)]$, 
%we shall write compactly the received signal during the training period as 
%$\bfY= \bfA \bfW + \bfN$. 
%
%From \eqref{eq:disc_ch_mod}, the covariance matrix of $\bfy [t]$ is given by
%\begin{align}\label{eq:subs_embed}
%{\bfC}_y=\bfA \mathbf{\Sigma} \bfA^\herm + \sigma^2 \Id_M = \sum_{\ell=1}^p \sigma_\ell^2 \bfa(\theta_\ell) \bfa(\theta_\ell)^\herm + \sigma^2 \Id_M.
%\end{align}
%We define the acquisition signal-to-noise ratio (SNR) as 
%\begin{align}\label{eq:snr_disc}
%\snr=\frac{\tr(\bfA \Sigmam\bfA^\herm)}{\tr(\sigma^2 \bfI_M)}=\frac{\sum_{\ell=1}^p\sigma_\ell^2}{\sigma^2}.
%\end{align}
For a more practical scenario, we consider the following continuum model 
\begin{align}\label{cont_sig_model}
\bfy [t]=\int_{-1}^1  \sqrt{\gamma(u)} \bfa(u) z(u, t) du+ \bfn [t], \ t \in [T],
\end{align}
where $z(u,t)$ is a circularly symmetric Gaussian process with the covariance
$
\bE\Big \{z(u,t) z(u',t')^* \Big \}= \delta(u - u') \delta_{t,t'}
$, and where $\gamma(u)$ is a positive measure that models the distribution of the received signal's power over $u\in [-1,1)$, where $u=\frac{\sin(\theta)}{\sin(\theta_m)}$ for $\theta \in [-\theta_m, \theta_m]$. With some abuse of notation, we denote the array vector in the $u$ domain by $\bfa(u)$ where $[\bfa(u)]_k=\exp(jk\pi u)$. It is not difficult to check that the covariance matrix of $\bfy [t]$ is given by 
\begin{align}
\bfC_y=\bfC(\gamma):=\int_{-1}^1 \gamma(u) \bfa(u) \bfa(u)^\herm du + \sigma^2 \bfI_M,
\end{align}
 where $\bfS(\gamma):=\int_{-1}^1 \gamma(u) \bfa(u) \bfa(u)^\herm du$ denotes the covariance of the signal part, and where $\sigma^2 \bfI_M$ is the covariance matrix of the white noise of the array. Note that the first column of $\bfS(\gamma)$ is given by $\inp{\gamma}{\bfa(u)}=\int_{-1} ^1 \gamma(u) \bfa(u) du$, whose $k$-th component 
\begin{align}
\Big[\inp{\gamma}{\bfa(u)}\Big]_k=\int_{-1} ^1 \gamma(u) \exp(j k\pi u) du,
\end{align}
 gives the $k$-th Fourier coefficient of the measure $\gamma(u)$. We denote these coefficients with an $M$-dim vector $\bff$, where $[\bff]_k=[\inp{\gamma}{\bfa(u)}]_k$ gives the $k$-th Fourier coefficient of $\gamma$. Also, note that for the uniform linear array (ULA), $\bfS(\gamma)$ is a Hermitian Toeplitz matrix with $[\bfS(\gamma)]_{r,c}=\bff_{r-c}$ for $r,c \in \{1,2,\dots, M\}$ with $r\geq c$. We also define the acquisition \textit{signal-to-noise ratio} (SNR) by $\snr=\frac{\tr(\bfS(\gamma))}{\tr(\sigma^2 \bfI_M)}=\frac{\bff_0}{\sigma^2}$.
%\begin{align}\label{eq:snr_cont}
%\snr=\frac{\tr(\bfS(\gamma))}{\tr(\sigma^2 \bfI_M)}=\frac{\bff_0}{\sigma^2},
%\end{align}
%which coincides with \eqref{eq:snr_disc} in the case of discrete angle-of-arrivales (AoA).

\subsection{Performance Metric}
Let $\bfC_y = \bfU_y \Lambdam \bfU_y^\herm$ be the {\em singular value decomposition} (SVD) of $\bfC_y$, where $\Lambdam =\diag(\lambda_1, \lambda_2, \dots, \lambda_M)$ denotes the diagonal matrix of singular values. We always assume that the singular values are sorted in a non-increasing order. 
For a given $p$, let us denote 
by $\bfU_{y,p}$ the $M\times p$ matrix consisting of the first $p$ columns of $\bfU_y$. Note that since the array noise is white, $\bfU_{y,p}$ can be interpreted as the best $p$-dim beamforming matrix that captures the highest amount of signal power. More precisely, 
\begin{align}\label{dom_subs}
\bfU_{y,p}=\argmax _{\bfU \in \bH(M,p)} \inp{\bfC_y}{\bfU\bfU^\herm},
\end{align}
where $\bH(M,p)$ denotes the space of all tall $M\times p$ matrices $\bfU$ with $\bfU^\herm \bfU=\bfI_p$. We asses the efficiency of the best $p$-dim beamformer by $\eta_p=\frac{\inp{\bfC_y}{\bfU_{y,p}\bfU_{y,p}^\herm}}{\tr(\bfC_y)}$. For a mm-Wave channel, $\eta_p \approx 1$ for a typically very small $p \ll M$ since a significant amount of signal's power is concentrated in a very low-dimensional subspace. Our goal in this paper is to reliably estimate the best $p$-dim beamformer $\bfU_{y,p}$ via an adaptive beamforming. Let $\widehat{\bfU}_p$ be such an estimate. We asses the efficiency of our estimator by $\Gamma_p=\frac{\inp{\bfC_y}{\widehat{\bfU}_p\widehat{\bfU}_p^\herm}}{\inp{\bfC_y}{\bfU_{y,p}\bfU_{y,p}^\herm}}$. Note that $\Gamma_p \in [0,1]$, where $\Gamma_p \approx 1$ implies that the estimate $\widehat{\bfU}_p$ is as good as the optimal beamformer $\bfU_{y,p}$.

\subsection{Signal Acquisition via Adaptive Beamforming}
At the start of the training period, the covariance matrix of the user $\bfC_y$ is unknown, thus, an acquisition step is necessary to estimate this matrix. Once we have an good estimate of $\bfC_y$, we can use its $p$-dim subspace $\bfU_{y,p}$ given by \eqref{dom_subs} for beamforming. We assume that we have access to the noisy output of the array $\bfy [t]$ for a training period of length $T$. In this paper, we focus on a single-RF-chain acquisition in which in each snapshot $t \in [T]$, we sample $\bfy[t]$ via an individual beamforming vector $\bfv [t]$ with $\|\bfv [t]\|=1$. This can be implemented via a single analog RF chain. We consider a noncoherent measurement given by\begin{align}\label{eq:noncoh_meas}
r [t]=\big |\inp{\bfv [t]}{\bfy [t]}\big |^2=\big |\bfv [t]^\herm \bfy [t]\big |^2.
\end{align} 
This measurement model is more robust than coherent measurements since during acquisition, the input receiver is not probably fully synchronized.
The beamforming vector $\bfv[t]$ at snapshot $t$ can be adaptively selected based on all the previous observations $r[\ell]$, $\ell=1,\dots, t-1$. 
From the signal model \eqref{cont_sig_model}, it is seen that $\inp{\bfv [t]}{\bfy [t]}$ is a complex Gaussian variable with a variance 
\begin{align}
\bfv [t]^\herm \bfC(\gamma) \bfv [t]=\sigma^2 + \bfv [t]^\herm \bfS(\gamma) \bfv [t],
\end{align}
which implies that $r [t]$ is an exponential random variable with a mean $\mu(\bfv [t])$ where $\mu: \bC^M \to \bR_+$ is given by $\mu(\bfv):=\sigma^2+ \bfv^\herm \bfS \bfv$, where for simplicity we drop the dependence of $\bfS(\gamma)$ on the power density $\gamma$. The variance of $r[t]$ is also given by $\kappa(\bfv[t])=\mu(\bfv[t])^2$. We always assume that $\bfy [t]$, and as a result $r[t]$, are independent across different snapshots $t\in[T]$, where the probability distribution of $r[t]$ is given by 
\begin{align}\label{exp_dist}
p(r)=\frac{1}{\mu(\bfv[t])}\exp\{-\frac{r}{\mu(\bfv[t])}\} \ind_{\{r\in \bR_+\}}.
\end{align}
For a fast and reliable acquisition, we need to find an adaptive procedure for designing the beamforming vectors $\bfv[t]$, $t \in [T]$ adaptively. The hope is that by using the previous observations $r[1], r[2], \dots, r[t-1]$, one can take more informative measurements from the signal covariance matrix.

%\subsection{Performance Metric}
%Let $\bfC_y = \bfU \Lambdam \bfU^\herm$ be the singular value decomposition (SVD) of $\bfC_y$, where $\Lambdam =\diag(\lambda_1, \lambda_2, \dots, \lambda_M)$ denotes the diagonal matrix of singular values. We always assume that the singular values are sorted in a non-increasing order. 
%For a given $p$, let us denote 
%by $\bfU_p$ the $M\times p$ matrix consisting of the first $p$ columns of $\bfU$. Note that $\bfU_p$ can be interpreted as the best $p$-dim beamforming matrix that capture the highest amount of signal power. More precisely, 
%\begin{align}\label{dom_subs}
%\bfU_p=\argmax _{\bfU \in \bH(M,p)} \inp{\bfC_y}{\bfU\bfU^\herm},
%\end{align}
%where $\bH(M,p)$ denotes the space of all tall $M\times p$ matrices $\bfU$ with $\bfU\bfU^\herm=\bfI_p$. 

\section{Problem Formulation}
\subsection{Basic setup}
We define the autocorrelation operator $\frka: \bC^M \to \bC^M$ that assigns to each vector $\bfv \in \bC^M$ the vector $\frka_\bfv\in \bC^M$ whose components are given by $[\frka_\bfv]_0=\|\bfv\|^2$ and 
\begin{align}\label{autocor}
[\frka_\bfv]_k=2\sum_{i=0}^{M-1-k} [\bfv]_{i+k} [\bfv]_i^\herm, \ k \in [M], k\geq 1.
\end{align}
Note that the first component of $\frka_\bfv$ is always real-valued.
%Since, we need to work with both real- and complex-valued inner products, we define some notation to simplify the formulation. 
We define the complex embedding $\frkc: \bR^{2M-1} \to \bC^M$ that maps
a vector $\bfr\in \bR^{2M-1}$ into the complex vector $\frkc_\bfr$ as follows:
\begin{align}
[\frkc_\bfr]_k&=[\bfr]_k, \ k=0,\\
[\frkc_\bfr]_k&=[\bfr]_k + j [\bfr]_{M-1+k},\ k\in [M], k\geq 1,
\end{align}
which can be also written with the shorthand notation
\begin{align}
\frkc_\bfr=\bfr[\range{0}{M}] + j \Big [0; \bfr[\range{M}{2M}]\Big].
\end{align}
Note that $[\frkc_\bfr]_0$ is still real-valued. We define the inverse map of $\frkc$ by $\frkr: \bC^M \to \bR^{2M-1}$, where for a $\bfc \in \bC^M$ with $[\bfc]_0 \in \bR$ 
\begin{align}
\frkr_\bfc=\Big [ \Re\big [ \bfc[\range{0}{M}] \big ] ; \Im\big [ \bfc[\range{M+1}{2M}]\big ] \Big ].
\end{align}
With the notation already defined, is not difficult to check that for any two real-valued signals $\bfu, \bfv \in \bR^{2M-1}$, we have
\begin{align}
\inpr{\frkc_\bfu}{\frkc_\bfv}:=\Re[\frkc_\bfu^\herm\, \frkc_\bfv]=\inp{\bfu}{\bfv}.
\end{align}
Now consider the exponential probability distribution in \eqref{exp_dist}, whose mean at time $t\in[T]$ is given by $\mu(\bfv[t])=\sigma^2 + \bfv[t]^\herm \bfS \bfv[t]$, where $\bfS$ is the covariance of the signal part. $\bfS$ is a Hermitian Toeplitz matrix with diagonal elements $\bff$ as defined in Section \ref{sec:sig_model}, i.e., $\bfS=\bT(\bff)$.
%It is not difficult to check that
%$\mu(\bfv[t])= \sigma^2 + \Re[\frka(\bfv[t])^\herm \bff]$. 
%
%Since, we are going to use real- and complex-valued inner products, we define some notation to simplify the formulation. We first define the embedding $\frkc: \bR^{2M-1} \to \bC^M$ which transforms an $(2M-1)$-dimensional signal $\bfr$ into an $M$-dimensional complex signal $\bfc=\frkc_\bfr$ given by:
%\begin{align}
%[\bfc]_k&=[\bfr]_k, \ k=1,\\
%[\bfc]_k&=[\bfr]_k + j [\bfr]_{M-1+k},\ k=2,3,\dots, M.
%\end{align}
%
%\begin{align}
%\frkc_\bfr=\bfr[\range{1}{M}] + j [0; \bfr[\range{M+1}{2M+1}];
%\end{align}
%Similarly, for an $M$-dimensional complex vector $\bfc$ with a real-valued first component, we define the reverse embedding as $\bfr=\frkr_\bfc$ given by
%\begin{align}
%[\bfr]_k&=[\bfc]_k, k=1\\
%[\bfr]_{k}&=\Re([\bfc]_k), \ k=2,3,\dots, M\\
%[\bfr]_{k}&=\Re([\bfc]_{k-M-1}), \ k=M+1, \dots, 2M-1.
%\end{align}
%With this notation, it is not difficult to check that for any two real-valued signals $\bfu, \bfv \in \bR^{2M-1}$, we have
%\begin{align}
%\inpr{\frkc(\bfu)}{\frkc(\bfv)}:=\Re[\frkc(\bfu)^\herm \frkc(\bfv)]=\inp{\bfu}{\bfv}.
%\end{align}
Using \eqref{autocor}, it is not difficult to check that, for every $\bfv, \bff \in \bC^M$ with $[\bff]_0 \in \bR$, 
\begin{align}\label{eq:con_relation}
\bfv^\herm \bT(\bff) \bfv= \inpr{\frka_\bfv}{\bff}=\inp{\frkr_{\frka_\bfv}}{\frkr_{\bff}}.
\end{align}
Also note that since $[\frka_\bfv]_0=\|\bfv\|^2$, and $[\bff]_0$ are real-valued, $\frkr_{\frka_\bfv}$ and $\frkr_\bff$ are well-defined. Moreover,
\begin{align}\label{eq:mu_real}
\mu(\bfv[t])= \sigma^2 + \bfv[t]^\herm \bT(\bff) \bfv[t]= \sigma^2 + \inp{\frkr_{\frka_\bfv[t]}}{\frkr_\bff}.
\end{align}  
From \eqref{exp_dist}, the Fischer Information Matrix (FIM) of the parameter $\bff$ (or equivalently $\frkr_\bff$) is given by
\begin{align}\label{FIM_t}
\frkF[t]&=\bE\big [\frac{\partial \log(p(r))}{\partial \frkr_\bff}  \frac{\partial \log(p(r))}{\partial \frkr_\bff}^\transp\big ]\\
&= {\frkr}_{\frka_\bfv[t]}{\frkr}_{\frka_\bfv[t]}^\transp \bE \Big[ \big (\frac{r}{\mu(\bfv[t])^2} - \frac{1}{\mu(\bfv[t])} \big )^2 \Big ]\\
&= {\frkr}_{\frka_\bfv[t]}{\frkr}_{\frka_\bfv[t]}^\transp \var[\frac{r}{\mu(\bfv[t])^2}]\\
%&=\frac{1}{\big \{\sigma^2+\inpr{\frka_\bfv[t]}{\bff}\big \}^2} {\frkr}_{\frka_\bfv[t]}{\frkr}_{\frka_\bfv[t]}^\transp\\
&=\frac{{\frkr}_{\frka_\bfv[t]}{\frkr}_{\frka_\bfv[t]}^\transp}{\big \{\sigma^2+\inp{\frkr_{\frka_\bfv[t]}}{\frkr_\bff}\big \}^2} .
\end{align}
where we used $\var[r]=\kappa(\bfv[t])=\mu(\bfv[t])^2$, and the $\mu(\bfv[t])$ given by \eqref{eq:mu_real}.
%Note that since the first component of $\frka_\bfv[t]$ is also real-valued, we can simply check that
%\begin{align}
%{\frkr}_{\frka_\bfv[t]}^\herm {\frkr}_\bff=\inpr{\frka_\bfv[t]}{\bff}=\bfv ^\herm \bT(\bff) \bfv.
%\end{align}
Since the observations $r[t]$ for different $t \in [T]$ are independent of each other, the FIM is additive, i.e.,
\begin{align}
\frkF[\range{1}{t}]=\sum_{i=1}^t \frkF[i].
\end{align}
We also define $\frkD[\range{1}{t}]:=\frkF[\range{1}{t}]^{-1}$. From the Cramer-Rao lower bound, the covariance matrix of any unbiased estimator $\widehat{\frkr_\bff}$ of the parameter $\frkr_\bff$ using the observations $r[\range{1}{t}]$ satisfies 
$\cov(\widehat{\frkr_\bff}) \succeq \frkD[\range{1}{t}]$. 
%, which depends on the choice of the beamforming vectors $\bfv[i]$ for $i=1,2,\dots, t$. 

\subsection{Design Criterion}
For the moment, suppose that the parameter $\bff$, and as a result the signal covariance matrix $\bT(\bff)$ is known. Note that from \eqref{FIM_t}, the FIM is PSD, which implies that $\frkD[\range{1}{t}] \succeq \frkD[\range{1}{t+1}]$, thus, Cramer-Rao lower bound always improves by adding new measurements. 
Let $\{\bfv[i]\}_{i=1}^t$ be the set of all beamforming vectors used until time $t$, and let
\begin{align}
\xi_t:=\frac{1}{t} \sum_{i=1}^t \delta[\bfw-\frkr_{\frka_\bfv[i]}],
\end{align}
where for a vector $\bfw_0\in \bR^{2M-1}$, we denote by $\delta[\bfw-\bfw_0]$ a delta-measure at point $\bfw_0$. Using $\xi_t$, we can write 
\begin{align}\label{fim_meas}
\frkF[\range{1}{t}]=t \int \frac{\bfw \bfw^\transp}{\big \{\sigma^2 + \inp{\bfw}{\frkr_\bff}\big \}^2}\, \xi_t(d\bfw),
\end{align}
which is a linear function of the design measure $\xi_t$. As we explained in Section \ref{sec:rel_work}, there are different optimality criterion that have been proposed. In this paper, we focus on a D-optimal design. For a nonadaptive design with at most $N$ samples, this requires finding a discrete measure with a discrete support of size at most $N$ that minimizes $\sdet(\frkD[\range{1}{N}])$. Relaxing the discreteness condition, this is a concave maximization that can be efficiently solved \cite{fedorov1972theory, fedorov2013optimal}.

In an adaptive design, we update the measure $\xi_{t-1}$ at time $t$ by adding a new design vector $\bfv[t]$. More precisely, 
\begin{align}
\xi_t= (1-\frac{1}{t}) \xi_{t-1}+ \frac{1}{t} \delta[\bfw - \frkr_{\frka_\bfv[t]}].
\end{align}
Hence, the FIM is updated as follows 
\begin{align}
\frkF[\range{1}{t}]=(1-\frac{1}{t}) \frkF[\range{1}{t-1}] + \frac{1}{t} \frkF[t].
\end{align}
For a D-optimal design, we minimize $\log(\sdet(\frkD[\range{1}{t}]))$. For sufficiently large $t$, the term $\frac{1}{t} \frkF[t]$ is a small rank-1 update of $\frkF[\range{1}{t-1}]$. Hence, we have
\begin{align}
\log(\sdet(\frkF[\range{1}{t}]))&= \log\Big(\sdet\big((1-1/t)\frkF[\range{1}{t-1}]\big)\Big)\\
 &+ \frac{1}{t-1} \frac{\frkr_{\frka_\bfv[t]}^\transp \frkD[\range{1}{t-1}] \frkr_{\frka_\bfv[t]}}{\big \{\sigma^2 + \inp{\frkr_{\frka_\bfv[t]}}{\frkr_\bff}\big \}^2} + o(\frac{1}{t})\nonumber.
\end{align}
This implies that for sufficiently large $t$, the best D-optimal design for $\bfv[t]$ at stage $t$ is given by
\begin{align}\label{opt_beamf}
\bfv^*[t]=\argmax_{\bfv \in \bC^M, \|\bfv\|=1} \frac{\frkr_{\frka_\bfv}^\transp \frkD[\range{1}{t-1}] \frkr_{\frka_\bfv}}{\big \{\sigma^2 + \inp{\frkr_{\frka_\bfv}}{\frkr_\bff}\big \}^2}
\end{align}
In practice, the parameter $\bff$ is unknown and we need to estimate it as well. Hence, we need to apply an iterative procedure, which has been summarized in {\bf Algorithm} \ref{adap_beam}.
%\begin{enumerate}
%\item Find the ML estimate $\widehat{\bff}[t]$ using the data $r[\range{1}{t-1}]$.
%\item Find the optimal beamforming vector $\bfv[t]$ via 
%\begin{align}
%\bfv[t]=\argmax_{\bfv \in \bC^M, \|\bfv\|=1} \frac{\frkr_{\frka_\bfv}^\transp \frkD[\range{1}{t-1}] \frkr_{\frka_\bfv}}{\big \{\sigma^2 + \inp{\frkr_{\frka_\bfv}}{\frkr_{\widehat{\bff}[t]}}\big \}^2},
%\end{align}
%where $\frkD[\range{1}{t-1}]= \frkF[\range{1}{t-1}]^{-1}$ is the inverse of the FIM  for the data $r[\range{1}{t-1}]$ computed at the estimated parameter value $\widehat{\bff}[t]$, i.e., 
%\begin{align}
%\frkF[\range{1}{t-1}]=\sum_{\ell=1}^{t-1} \frac{{\frkr}_{\frka_\bfv[\ell]}{\frkr}_{\frka_\bfv[\ell]}^\transp}{\big \{\sigma^2+\inp{\frkr_{\frka_\bfv[\ell]}}{\frkr_{\widehat{\bff}[t]}}\big \}^2}.
%\end{align}
%\item Take the measurement at time $t$ via the beamforming vector $\bfv[t]$ to obtain the data $r[t]$. Update the observations to $r[\range{1}{t}]=r[\range{1}{t-1}] \cup \{r[t]\}$.
%%, and the FIM to $\frkF[t]= (1-\frac{1}{t}) \frkF[t-1] + \frac{1}{t} \frac{{\frkr}_{\frka_\bfv[t]}{\frkr}_{\frka_\bfv[t]}^\transp}{\big \{\sigma^2+\inp{\frkr_{\frka_\bfv[t]}}{\frkr_{\widehat{\bff}[t]}}\big \}^2}$
%\end{enumerate} 

\begin{algorithm}[h!]
\caption{Adaptive Beamforming and Subspace Estimation}
\label{adap_beam}
\begin{algorithmic}[1]
\State {\color{red}{\bf Initial Sampling:}}
\For{$t=1,\dots, 2M-1$}
\State Choose a random beamforming vector $\bfv[t]$.
\State Take a noncoherent measurement $r[t]=|\inp{\bfv[t]}{\bfy[t]}|^2$.
\EndFor
\State $t \gets 2M$
\State {\color{red}{\bf Sequential Acquisition and Subspace Estimation:}}
\While{allowed to take further measurements}
\State Find the ML estimate ${\bff}^*[t]$ using the data $r[\range{1}{t-1}]$. 
\State Compute $\frkF[\range{1}{t-1}]$ and $\frkD[\range{1}{t-1}]$ for $\bff=\bff^*[t]$.
\State \begin{varwidth}[t]{\linewidth}
Find the optimal beamforming vector $\bfv^*[t]$ via \eqref{opt_beamf}\\ for $\bff=\bff^*[t]$.
\end{varwidth}
\State Take a new sample $r[t]=|\inp{\bfv^*[t]}{\bfy[t]}|^2$ using $\bfv^*[t]$.
\State Update the observations $r[\range{1}{t}]=r[\range{1}{t-1}] \cup \{r[t]\}$.
\State $t \gets t+1$.
\EndWhile
\end{algorithmic}
\end{algorithm}

\section{Analysis and Proof Techniques}
\subsection{Maximum Likelihood Estimate of $\bff$}\label{sec:ML}
Using the independence of $r[\range{1}{t}]$, from \eqref{exp_dist} we can write their probability distribution as follows
\begin{align}
p(r[\range{1}{t}])=\prod_{\ell=1}^t \frac{1}{\mu(\bfv[\ell])}\exp\{-\frac{r[\ell]}{\mu(\bfv[\ell])}\} \ind_{\{r[\ell] \in \bR_+\}}.
\end{align}
The ML estimator is found by minimizing the cost function
\begin{align*}
L_t(\bff)=\sum_{\ell=1}^t \log(\sigma^2 + \inp{\frkr_{\frka_\bfv[\ell]}}{\frkr_\bff}) + \sum_{\ell=1}^t \frac{r[\ell]}{\sigma^2 + \inp{\frkr_{\frka_\bfv[\ell]}}{\frkr_\bff}}
\end{align*}
subject to the constraint $\bT(\bff) \in \bT_+$. 
It is seen that $L_t(\bff)$ is the sum of a logarithmic concave term and a convex one, thus, it is not in general convex. However, a simple calculation shows that the Hessian of $L_t(\bff)$ is given by 
\begin{align}
\nabla^2 L_t(\bff)= \sum_{\ell=1}^t \frac{\frkr_{\frka_\bfv[\ell]}\frkr_{\frka_\bfv[\ell]}^\transp}{\mu(\bfv[\ell])^2} \big(r[\ell] -1 + \frac{r[\ell]}{\mu(\bfv[\ell])} \big ).
\end{align}
whose expectation is given by 
\begin{align}
\bE\big[\nabla^2 L_t(\bff)\big]= \sum_{\ell=1}^t \frac{\frkr_{\frka_\bfv[\ell]}\frkr_{\frka_\bfv[\ell]}^\transp}{\mu(\bfv[\ell])}=\sum_{\ell=1}^t \frac{\frkr_{\frka_\bfv[\ell]}\frkr_{\frka_\bfv[\ell]}^\transp}{\sigma^2 + \inp{\frkr_{\frka_\bfv[t]}}{\frkr_\bff}},
\end{align}
which is a PSD matrix. This implies that for sufficiently large $t$, we expect that $L_t(\bff)$ be a convex function of $\bff$, which can be efficiently minimized over the convex set $\bT_+$. In this paper, we apply a \textit{concave-convex procedure} (CCCP) first proposed in \cite{yuille2003concave} to sequentially estimate the 
optimal solution $\bff^*[t]$. Let $\{\bff_\ell\}_{\ell=0}^{k}$ be a sequence of estimate solutions obtained via CCCP with the initialization $\bff_0=(1,0,0, \dots, 0)^\transp$. To find the next estimate $\bff_{k+1}$, we first upper bound 
the logarithmic term $\lcav(\bff):=\sum_{\ell=1}^t \log(\sigma^2 + \inp{\frkr_{\frka_\bfv[t]}}{\frkr_\bff})$ in function $L_t(\bff)$ around the latest estimate $\bff_{k}$ by 
the linear term 
\begin{align*}
\ell(\bff;\bff_{k}):= \lcav(\bff_k)+\sum_{\ell=1}^t 
\frac{\inp{\frkr_{\frka_\bfv[t]}}{\frkr_{\bff-\bff_k}}}{\sigma^2+\inp{\frkr_{\frka_\bfv[t]}}{\frkr_{\bff_k}}},
\end{align*}
which is the tightest convex approximation of the concave function $\lcav(\bff)$. In particular, $\ell(\bff_k;\bff_k)=\lcav(\bff_k)$. Let $\Upsilon(\bff;\bff_k)= \sum_{\ell=1}^t \frac{r[\ell]}{\sigma^2 + \inp{\frkr_{\frka_\bfv[t]}}{\frkr_\bff}} + \ell(\bff;\bff_k)$. It is not difficult to check that for every $\bff\in\bT_+$, the function $\Upsilon(\bff;\bff_k)$ is a convex upper bound for $L_t(\bff)$. We obtain the next estimate $\bff_{k+1}$ via the following convex optimization
\begin{align}
\bff_{k+1}:=\argmin_{\{\bff: \bT(\bff) \in \bT_+\}} \Upsilon(\bff;\bff_k).
\end{align}
We can check that $\{\bff_\ell\}_{\ell=1}^\infty$ monotonically improves the likelihood function $L_t(\bff)$ since
\begin{align}
L_t(\bff_{k+1})&\leq \Upsilon(\bff_{k+1};\bff_k)=\min_{\{\bff: \bT(\bff) \in \bT_+\}} \Upsilon(\bff;\bff_k)\\
&\leq \Upsilon(\bff_k;\bff_k)=L_t(\bff_{k}).
\end{align}
In particular, if $L_t(\bff)$ happens to be convex then every limit point of the sequence $\{\bff_\ell\}_{\ell=1}^\infty$ generated by the CCCP will correspond to the globally optimal solution $\bff^*[t]:=\argmin_{\{\bff: \bT(\bff) \in \bT_+\}} L_t(\bff)$. 
%
%
%estimate the concave logarithmic term   In general, since $L_t(\bff)$ is not convex, we can approximate it by the convex upper bound
%\begin{align}\label{conv_app}
%\overline{L}_t(\bff)=\sum_{\ell=1}^t \frac{\inp{\frkr_{\frka_\bfv[t]}}{\frkr_\bff}}{\sigma^2} + \frac{r[\ell]}{\sigma^2 + \inp{\frkr_{\frka_\bfv[t]}}{\frkr_\bff}}
%\end{align}
%that is obtained by approximating the concave term $\log(\sigma^2+\inp{\frkr_{\frka_\bfv[t]}}{\frkr_\bff})$ by a convex (indeed linear) upper bound $\log(\sigma^2) + \frac{\inp{\frkr_{\frka_\bfv[t]}}{\frkr_\bff}}{\sigma^2}$. Optimization \eqref{conv_app} provides a first estimate of the optimal solution. This estimate can be improved by another  step of gradient descent.

\subsection{Finding the Optimal Beamforming Vector}
For finding the optimal beamforming vector at stage $t$, we need to solve the optimization problem \eqref{opt_beamf}, which can be equivalently written as  
\begin{align}\label{opt_beamf2}
\bfv[t]=\argmax_{\bfv \in \bC^M, \|\bfv\|=1} \frac{\frkr_{\frka_\bfv}^\transp \frkD[\range{1}{t-1}] \frkr_{\frka_\bfv}}{\inp{\frkr_{\frka_\bfv}}{\frkr_{\widehat{\bff}[t] + \sigma^2 \bfe_1}}^2},
\end{align}
where $\bfe_1=(1,0, \dots, 0)^\transp\in \bR^{2M-1}$ is the first element of standard basis in $\bR^{2M-1}$. To simplify the optimization \eqref{opt_beamf2}, we define
\begin{align*}
\clA=\{\bfr\in \bR^{2M-1}: \exists\, \bfv\in \bC^M \text{ such that } \|\bfv\|=1, \bfr=\frkr_{\frka_\bfv}\},
\end{align*}
as the set of autocorrelation sequences. Note that for every $\bfr \in \clA$, we have $[\bfr]_0=1$ and from \eqref{eq:con_relation}
\begin{align}
[\bfr]_0 &+ \sum_{k=1}^{M-1} [\bfr]_k \cos(k \pi u) + [\bfr]_{k+M} \sin(k \pi u)\\
&=\inp{\bfr}{\frkr(\bfa(u))}=\inpr{\frkc_\bfr}{\bfa(u)}= \inpr{\frka_\bfv}{\bfa(u)}\\
&= \bfv^\herm \bT(\bfa(u))\bfv=|\bfv^\herm \bfa(u)|^2\geq 0,
\end{align}
where $\bfv$ is a vector whose autocorrelation is $\bfr$. It is seen that for every $\bfr\in \clA$ the trigonometric function $\inpr{\frkc_\bfr}{\bfa(u)}$ is positive in the whole interval $u \in [-1,1]$. In fact, from the Riesz-Fej\'er spectral factorization the converse also holds \cite{ben2001lectures, dumitrescu2007positive}, i.e., a vector $\bfr \in \clA$ if and only if $[\bfr]_0=1$ and for every $u \in [-1,1]$ the trigonometric function $\inpr{\frkc_\bfr}{\bfa(u)}$ is positive for all $u\in [-1,1]$. This in particular implies that $\clA$ can be written as the intersection of closed half-spaces 
\begin{align}\label{inf_rep}
\clA=\cap_{u \in [-1,1]} \{\bfr: [\bfr]_0=1, \inpr{\frkc_\bfr}{\bfa(u)}\geq 0\},
\end{align}
which implies that $\clA$ is a closed convex set. Moreover, from the properties of the autocorrelation sequence, it results that for every $1\leq k \leq M-1$, we have $|\bfr_k|\leq 2\bfr_0=2$, which implies that $\clA$ is also bounded, thus, it is a compact subset of $\bR^{2M-1}$. Hence, we can write the optimization \eqref{opt_beamf2} in the following equivalent form:
\begin{align}\label{opt_beamf3}
\bfr^*=\argmax_{\bfr \in \clA} \frac{\bfr^\transp \bfD \bfr}{\inp{\bfr}{\hat{\bfr}}^2},
\end{align}
where $\bfD:=\frkD[\range{1}{t-1}]$ and $\hat{\bfr}:=\frkr_{\widehat{\bff}[t] + \sigma^2 \bfe_1}$. Let $\bfr \in \clA$ be the autocorrelation of some $\bfv \in \bC^M$, i.e., $\bfr=\frkr_{\frka_\bfv}$ with $\|\bfv\|=1$. From \eqref{eq:con_relation}, we have 
\begin{align}
\inp{\bfr}{\hat{\bfr}} &=\inp{\frkr_{\frka_\bfv}}{\frkr_{\widehat{\bff}[t] + \sigma^2 \bfe_1}}=\inpr{\frka_\bfv}{\widehat{\bff}[t] + \sigma^2 \bfe_1}\\
&=\bfv^\herm \bT(\widehat{\bff}[t] + \sigma^2 \bfe_1)\bfv \geq \sigma^2 \|\bfv\|^2=\sigma^2>0,
\end{align}
thus, the cost function in  \eqref{opt_beamf3} is well-defined. Moreover, since $\clA$ is compact the maximum is also achieved.
The cost function in \eqref{opt_beamf3} is the ratio of two convex quadratic functions, which can be optimized by  a bisection procedure first proposed by Dinkelbach \cite{dinkelbach1967nonlinear}. Let $\lambda \in \bR_+$ and define 
\begin{align}
Q_\lambda(\bfr)&=\bfr^\transp \bfD \bfr-\lambda \inp{\bfr}{\hat{\bfr}}^2=\bfr^\transp \bfD_\lambda \bfr,
\end{align}
where $\bfD_\lambda :=\bfD -\lambda \hat{\bfr}\hat{\bfr}^\transp$. Let $\pi(\lambda)=\max_{\bfr \in \clA} Q_\lambda(\bfr)$. Note that since $\clA$ is compact the maximum is always achieved for some $\bfr^*(\lambda) \in \clA$. We have the following proposition.
\begin{proposition}\label{pi_sol}
Let $\lambda \in \bR_+$ and let $\pi(\lambda)$ be as defined before. Then, $\pi(\lambda)$ is a decreasing convex function of $\lambda$. Moreover, there is some $\lambda^* \in \bR_+$ such that $\pi(\lambda^*)=0$. 
\end{proposition}
\begin{proof}
The decreasing property follows from the definition and the fact that $\bfD_{\lambda_1} \succeq \bfD_{\lambda_2}$ for every $\lambda_2 > \lambda_1 \geq 0$. To prove the convexity, let $\lambda_1, \lambda_2$ be arbitrary numbers in $\bR_+$ and let $\alpha \in (0,1)$ and set $\lambda_\alpha=\alpha \lambda_1 + (1-\alpha) \lambda_2$. Let $\bfr_\alpha=\bfr^*(\lambda_\alpha)=\argmax_{\bfr\in \clA} Q_{\lambda_\alpha}(\bfr)$. Then, we have
\begin{align}
\pi(\lambda_\alpha)&=Q_{\lambda_\alpha}(\bfr_\alpha)= \alpha Q_{\lambda_1}(\bfr_\alpha) + (1-\alpha) Q_{\lambda_2}(\bfr_\alpha)\\
&\leq \alpha \pi(\lambda_1) + (1-\alpha) \pi(\lambda_2),
\end{align}
which proves the convexity of $\pi(\lambda)$. In particular, it results that $\pi(\lambda)$ is continuous in $\bR_+$. 

To prove the last part, let $\pi(0)=\max_{\bfr \in \clA} \bfr^T \bfD \bfr$. Since $\clA$ is compact $\pi(0) \in (0, \infty)$ is a bounded positive number. As $\inp{\bfr}{\hat{\bfr}}\geq \sigma^2$, it results that for every $\bfr \in \clA$, we have $\pi(\lambda) \leq \pi(0) - \lambda \sigma^4$ for all $\lambda \in \bR_+$. This implies that there is a unique  $\lambda^* \in \bR$ such that $\pi(\lambda^*)=0$. 
\end{proof}

The next result characterizes the global maximum of \eqref{opt_beamf3}.

\begin{proposition}\label{lam_sol}
Let $\pi(\lambda)$ be as before and let $\lambda^*$ be the unique point given by Proposition \ref{pi_sol} for which $\pi(\lambda^*)=0$. Then the global maxizer of \eqref{opt_beamf3} is given by $\bfr^*(\lambda^*)$.
\end{proposition}
\begin{proof}
From the definition of $\pi(\lambda)$, we have
\begin{align}
\bfr^\transp \bfD_{\lambda^*} \bfr \leq \pi(\lambda^*)=0,
\end{align}
for every $\bfr \in \clA$. In particular, $\bfr^\transp \bfD \bfr \leq \lambda^* \inp{\bfr}{\hat{\bfr}}^2$ or  $\frac{\bfr^\transp \bfD \bfr}{\inp{\bfr}{\hat{\bfr}}^2} \leq \lambda^*$, where the equality is achieved for $\bfr=\bfr^*(\lambda^*)$. This implies that $\bfr^*(\lambda^*)$ is the global maximizer of \eqref{opt_beamf3} over $\clA$ with a maximum value $\lambda^*$.
\end{proof}

\begin{figure*}[t!]
\centering
\begin{tikzpicture}
\pgfmathsetmacro{\scale}{0.68}
\pgfmathsetmacro{\shift}{5.8cm}
\pgfmathsetmacro{\mksz}{3pt}

\begin{scope}[scale=\scale]
\begin{axis}[ title={SNR $=0$ dB}, ylabel={Performance Metric $\Gamma_p$}, xmin=8, xmax=400, ymax=1.05, ymin=0, grid=major,  legend style={
        cells={anchor=east},
        legend pos= south east,
    }]

\pgfplotsset{every axis/.append style={
                    legend style={mark size=\mksz, style={nodes={right}}},
                    }}
\addplot+[mark size=\mksz, mark=diamond*]  table[x expr=\thisrowno{0}*10, x=iter, y=metric] {M20SNR0REP10.txt};
\addplot+[mark size=\mksz, mark=diamond]  table[x expr=\thisrowno{0}*10, x=iter, y=metric] {M20SNR0REP10_exh.txt};
\legend{Adaptive, Exhaustive};
\end{axis}
\end{scope}

\begin{scope}[xshift=\shift, scale=\scale]
\begin{axis}[xlabel={Time (Number of Sequential Training Samples)}, xmin=8, xmax=400, ymax=1.05, ymin=0, title={SNR $=-10$ dB}, yticklabels={}, grid=major,  legend style={
        cells={anchor=east},
        legend pos= south east,
    }]

\pgfplotsset{every axis/.append style={
                    legend style={mark size=\mksz, style={nodes={right}}},
                    }}
\addplot+[mark size=\mksz, mark=diamond*]  table[x expr=\thisrowno{0}*10, x=iter, y=metric] {M20SNR-10REP10.txt}; 
\addplot+[mark size=\mksz, mark=diamond]  table[x expr=\thisrowno{0}*10, x=iter, y=metric] {M20SNR-10REP10_exh.txt}; 
\legend{Adaptive, Exhaustive};
\end{axis}
\end{scope}

\begin{scope}[xshift=2*\shift, scale=\scale]
\begin{axis}[xmin=8, xmax=400, ymax=1.05, ymin=0, title={SNR $=-20$ dB}, yticklabels={}, grid=major,  legend style={
        cells={anchor=east},
        legend pos= north west,
    }]

\pgfplotsset{every axis/.append style={
                    legend style={mark size=\mksz, style={nodes={right}}},
                    }}
\addplot+[mark size=\mksz, mark=diamond*]  table[x expr=\thisrowno{0}*10, x=iter, y=metric] {M20SNR-20REP10.txt}; 
\addplot+[mark size=\mksz, mark=diamond]  table[x expr=\thisrowno{0}*10, x=iter, y=metric] {M20SNR-20REP10_exh.txt}; 
\legend{Adaptive, Exhaustive};
\end{axis}
\end{scope}
\end{tikzpicture}
\caption{Comparing the performance of our algorithm (Adaptive) with the exhaustive search (Exhaustive) for different SNR.}
\label{fig:perf_adap}
\end{figure*}
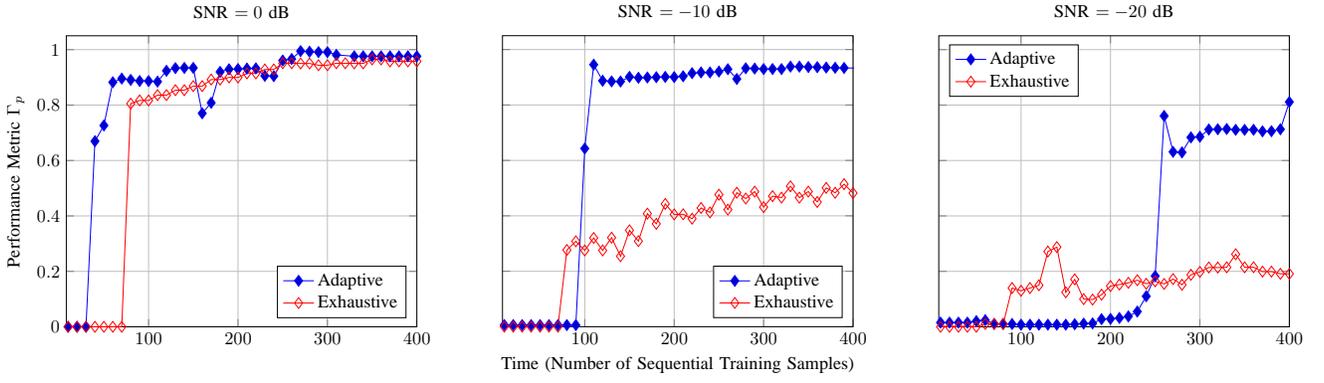
From Proposition \ref{pi_sol}, the function $\pi(\lambda)$ is decreasing and continuous, thus, we can find the optimal $\lambda^*$ via a bisection procedure. To do this we need to be able to find $\pi(\lambda)$ for every given $\lambda \in \bR_+$ from the maximization 
\begin{align}\label{lam_max_find}
\pi(\lambda)=\max_{\bfr \in \clA} \bfr^\transp(\bfD - \lambda \hat{\bfr} \hat{\bfr}^\transp) \bfr= \max_{\bfr \in \clA} \bfr^\transp \bfD_\lambda \bfr.
\end{align}
The main challenge is that although the set $\clA$ is convex, the quadratic form $Q_\lambda(\bfr)$ is in general non-concave (indefinite), and the maximization \eqref{lam_max_find} is difficult to do. It has been shown that by some change of variables, the maximization of every indefinite quadratic function in \eqref{lam_max_find} can be converted to a convex maximization \cite{pardalos1991global} that can be approximately solved using the global optimization techniques. But, there is still no guarantee that the resulting solution be globally optimal.

Similar to the ML estimation in Section \ref{sec:ML}, we apply the CCCP \cite{yuille2003concave} to find an approximate solution for \eqref{lam_max_find}. Let $\gamma >0$ be such that $\inp{\bfD - \gamma \hat{\bfr} \hat{\bfr}^\transp}{\hat{\bfr} \hat{\bfr}^\transp}=0$, where the inner product is the traditional matrix inner product. This implies that $\gamma=\frac{\inp{\hat{\bfr}}{\bfD \hat{\bfr}}}{\inp{\hat{\bfr}}{\hat{\bfr}}^2}$. We can decompose $\bfD -\gamma \hat{\bfr}\hat{\bfr}^\transp$ into the difference of two PSD matrices $\bfP:= \bfD - \gamma \hat{\bfr} \hat{\bfr}^\transp$ and $\bfN:=(\lambda+ \gamma) \hat{\bfr} \hat{\bfr}^\transp$, where we obtain $Q_\lambda(\bfr)= \inp{\bfr}{\bfP \bfr} - \inp{\bfr}{\bfN \bfr}$. Using the CCCP, we generate the following sequence of estimates for the optimal solution $\bfr^*$. We start with the initial point $\bfr_0=(1,0,\dots, 0)^\transp \in \bR^{2M-1}$. Suppose $\{\bfr_\ell\}_{\ell=0}^k$ is the sequence of estimates generated by CCCP. To find the next estimate $\bfr_{k+1}$, we first approximate the convex function $\inp{\bfr}{\bfP \bfr}$ by the linear term $\ell(\bfr;\bfr_k)=\inp{\bfr_k}{\bfP \bfr_k} + 2\inp{\bfP \bfr_k}{\bfr}$ around the latest estimate $\bfr_k$. Note that for every $\bfr \in \clA$, we have $\inp{\bfr}{\bfP \bfr} \geq \ell(\bfr;\bfr_k)$. We also define $\Upsilon(\bfr;\bfr_k):= \ell(\bfr;\bfr_k)- \inp{\bfr}{\bfN \bfr}$. Note that $\Upsilon(\bfr;\bfr_k)$ is a concave lower-bound for $Q_\lambda(\bfr)$ over $\clA$. We find the next estimate $\bfr_{k+1}$ by $\bfr_{k+1}=\argmax_{\bfr \in \clA} \Upsilon(\bfr;\bfr_k)$, which is a concave maximization that can be efficiently solved. It is not also difficult to show that for every $\bfr_k$ and $\bfr_{k+1}$ in the CCCP sequence, $Q_\lambda(\bfr_{k+1}) \geq Q_\lambda(\bfr_k)$, which implies that CCCP monotonically improves the cost function $Q_\lambda$.

Another difficulty for solving the optimization \eqref{lam_max_find} is that, from \eqref{inf_rep} the convex region $\clA$ is represented as the intersection of infinitely many half-spaces. In practice, we need to find a finite-dimensional approximation for $\clA$. One way is to approximate $\clA$ by a polyhedral set obtained as the intersection of finitely many subspaces given by 
\begin{align}\label{fin_rep}
\clA_{\Omega_p}=\cap_{u \in \Omega_p} \{\bfr: [\bfr]_0=1, \inpr{\frkc_\bfr}{\bfa(u)}\geq 0\},
\end{align}
where $\Omega_p=\{u_k\in [-1,1]: k=1,2,\dots, p\}$ is a finite grid of size $p$ in $[-1,1]$. 
In \cite{dumitrescu2007positive}, another approach has been given to approximate $\clA$ by embedding $\clA$ in higher-order Toeplitz matrices. More precisely, let $M'\geq M$ and define
\begin{align}
\clA_{M'}=\{\bfr \in \bC^M: [\bfr]_0=1, \bT([\bfr^\transp, {\bf 0}_{M'-M}]^\transp) \in \bT_+\}.
\end{align}
It is not difficult to see that $\clA_{M'}$ is a closed convex subset of $\bC^M$. Moreover, since the principal submatrices of a PSD matrix should be PSD as well, we have $\clA_{M'} \subset \clA_{M'+1}$ for every $M'\geq M$. In \cite{dumitrescu2007positive}, it has also been shown that $\clA= \clA_{\infty}=\cap_{M'=M}^\infty \clA_{M'}$, which is closed and convex as expected. In practice, we still need to use a finite but sufficiently large $M'\geq M$ to obtain a good approximation of $\clA$.

\section{Simulation Results}
For simulation, we consider an array of size $M=20$, and a user whose received power is uniformly distributed in the angular range $[-50,-48] \cup [10, 12]$ degrees. Fig.~\ref{fig:C_svd} shows the SVD of the covariance matrix of the user's signal. It is seen that a significant amount of user's signal power is concentrated in a subspace of dimension $2$.

Fig.~\ref{fig:perf_adap} shows the performance of our proposed algorithm for different training sample sizes and different pre-beamforming signal-to-noise ratios. 
We compare our algorithm with a trivial scheme that partitions the AoAs in $M$ equally-spaced bins, e.g. $\Delta_M=\{-\theta_m +\frac{2i\theta_m}{M}: i\in [M]\}$ for $\theta_m=\frac{\pi}{2}$, and in each step uses the array vector spanning one of these bins, i.e., $\bfa(\theta_b)$, $\theta_b \in \Delta_M$, to take a sample from the $M$-dimensional received signal in the array. The algorithm iteratively sweeps $\Delta_M$, and after each sweep finds the $2$ dominant angles with the highest received power.

\begin{figure}[t]
\centering
\begin{tikzpicture}[scale=0.55]
\pgfmathsetmacro{\mksz}{3pt}

\begin{axis}[ylabel={SVD}, , xmin=0.9, xmax=20, ymin=0, grid=major,  legend style={
        cells={anchor=east},
        legend pos= south east,
    }]

\pgfplotsset{every axis/.append style={
                    legend style={font=\scriptsize,mark size=\mksz},
                    }}
\addplot+[mark size=\mksz]  table[x=subspace, y=svd] {C_svd.txt};
\end{axis}
\end{tikzpicture}
\caption{SVD of signal's covariance matrix}
\label{fig:C_svd}
\end{figure}
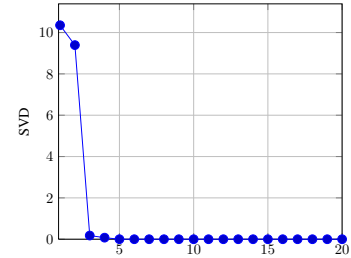

{\footnotesize
\bibliographystyle{IEEEtran}
\bibliography{paper.bbl}
}

\end{document}